\documentclass[reqno]{amsart}
\usepackage{amssymb, amsthm, amsmath}
\usepackage{graphicx}
\usepackage{enumitem}
\usepackage{hyperref}

\newtheorem{theorem}{Theorem}
\newtheorem{lemma}[theorem]{Lemma}

\newtheorem{proposition}[theorem]{Proposition}

\theoremstyle{definition}

\newtheorem{remark}[theorem]{Remark}

\theoremstyle{remark}

\newcommand{\EEACUE}{\mathbb{E}_{\mathrm{ACUE}(N)}}
\newcommand{\EECUE}{\mathbb{E}_{\mathrm{CUE}(N)}}
\newcommand{\ACUE}{\mathrm{ACUE}}
\newcommand{\CUE}{\mathrm{CUE}}

\newfont{\cmbsy}{cmbsy10}
\newfont{\cmmib}{cmmib10}

\renewcommand{\Re}{\operatorname{Re}}

\newcommand{\diag}{\mathrm{diag}}

\newcommand{\Mod}[1]{\ (\mathrm{mod}\ #1)}

\newcommand{\Addresses}{{
		\footnotesize
		\bigskip
		\footnotesize
		
		\textsc{Department of Mathematics and Statistics, Queen's University, Kingston, Ontario, K7L 3N6, Canada}\par\nopagebreak
		\textit{E-mail address:} \texttt{brad.rodgers@queensu.ca}
		
		\medskip
		
		\textsc{Indian Institute of Technology Kharagpur, Kharagpur, West Bengal 721302, India}\par\nopagebreak
		\textit{E-mail address:} \texttt{vallabhaneniharshith@gmail.com}
		
}}


    \title[Autocorrelations for the ACUE]{Autocorrelations of characteristic polynomials for the Alternative Circular Unitary Ensemble}
    \author{Brad Rodgers, Harshith Sai Vallabhaneni}
    \date{}
    
\begin{document}

\maketitle
\thispagestyle{empty}

\begin{abstract}
We find closed formulas for arbitrarily high mixed moments of characteristic polynomials of the Alternative Circular Unitary Ensemble (ACUE), as well as closed formulas for the averages of ratios of characteristic polynomials in this ensemble. A comparison is made to analogous results for the Circular Unitary Ensemble (CUE). Both moments and ratios are studied via symmetric function theory and a general formula of Borodin-Olshanski-Strahov.
\end{abstract}

\section{Introduction}
\label{sec:intro}

In this short note we examine mixed moments and averages of ratios of characteristic polynomials associated with the \emph{Alternative Circular Unitary Ensemble (ACUE)}. Our main results are a closed formula for arbitrarily high mixed moments in Theorem \ref{thm:main} and a closed formula for averages of ratios in Theorem \ref{thm:BOS}. The ACUE refers to a certain random collection of points on the unit circle of the complex plane whose distribution is meant to mimic the points of the Circular Unitary Ensemble (CUE) of random matrix theory. Let us use the notation
$$
\Delta(x_1,...,x_N):= \prod_{1 \leq j < k \leq N}(x_j-x_k),
$$
\sloppy for a Vandermonde determinant\footnote{Note that some authors define the Vandermonde determinant in such a way as to have the opposite sign, but we will be consistent with the notational convention above.}, an anti-symmetric function in the variables $x_1,...,x_N$, and let us also use the notation $e(t):= e^{i2\pi t}$. For an integer $N \geq 1$, we use the label $\ACUE(N)$ to denote the random collection of $N$ points $\{e(\vartheta_1),...,e(\vartheta_N)\}$ on the unit circle $S^1$ of the complex plane which have the following joint density: for an arbitrary measurable function $f: (S^1)^N \rightarrow \mathbb{C}$,
\begin{multline}
\label{eq:ACUEdef}
\EEACUE \big[f(e(\vartheta_1),...,e(\vartheta_N))\big] \\
= \frac{1}{N!} \frac{1}{(2N)^N}\sum_{t_1,...,t_N} f(e(t_1),...,e(t_N)) \cdot |\Delta(e(t_1),...,e(t_N))|^2,
\end{multline}
where each index $t_i$ is summed over the set $\{0,\tfrac{1}{2N}, \tfrac{2}{2N},...,\tfrac{2N-1}{2N}\}$ (so that the sum consists of $(2N)^N$ terms in total). Likewise we use the label $\CUE(N)$ to denote the random collection of $N$ points $\{e(\theta_1),...,e(\theta_N)\}$ on the unit circle with joint density given by:
\begin{multline}
\label{eq:CUEdef}
\mathbb{E}_{\mathrm{CUE}(N)} \big[f(e(\theta_1),...,e(\theta_N))\big]\\ 
= \frac{1}{N!} \int_{[0,1]^N} f(e(t_1),...,e(t_N)) \cdot |\Delta(e(t_1),...,e(t_N))|^2\, d^Nt.
\end{multline}
It is known (see \cite[Eq. (21)]{Ta}) that $\EEACUE [1] = \EECUE [1] = 1$ so both these expressions indeed implicitly define joint probability densities. These joint densities are each symmetric in all variables, so the $\ACUE(N)$ and the $\CUE(N)$ may be seen as point processes supported on the $2N$-th roots of unity or the unit circle of the complex plane respectively. We use the notation $\mathbb{E}_{\mathrm{ACUE}(N)}$ or $\mathbb{E}_{\mathrm{CUE}(N)}$ for the purpose of reminding the reader over which ensemble an expectation is being taken. (These could be replaced by the more traditional notation $\mathbb{E}$ with no change in meaning.)

The ACUE was put forward in a blog post of T. Tao \cite{Ta} in order to investigate the limitations of certain methods in analytic number theory. Of particular interest was a comparison of the $k$-level correlation functions of the ACUE and the CUE. The CUE can be seen as a finite model of how zeros of the Riemann zeta function and other L-functions are conjectured to be spaced, while the ACUE can be seen as a finite model of how zeros are very unlikely to be spaced but which cannot be ruled out by current methods. A similar construction (replacing the CUE and ACUE with limiting point processes) was independently studied by J. Lagarias and the first author of this paper \cite{LaRo1} around the same time.  Related point processes have also been studied for reasons unrelated to number theory in the past; see e.g. \cite{Bo,BoOkOl}.

It is therefore of interest to investigate similarities and differences between the CUE and ACUE. In this paper we examine the statistics induced by characteristic polynomials associated to the CUE and ACUE. The CUE is naturally associated to the eigenvalues of a random Haar distributed unitary matrix, but there is not an especially natural matrix interpretation for the ACUE (though see Remark 6 of \cite{Ta}). In order to easily speak of the characteristic polynomial associated to these ensembles, define the diagonal matrices
\begin{align*}
g&:= \diag(e(\vartheta_1),...,e(\vartheta_N)) \qquad &&\textrm{(associated to $\ACUE(N)$)}\\
G&:= \diag(e(\theta_1),...,e(\theta_N)) \qquad &&\textrm{(associated to $\CUE(N)$)}.
\end{align*}
We refer to the random functions $\det(1-zg)$ and $\det(1-zG)$ in the complex variable $z$ as \emph{the characteristic polynomials} associated with the ACUE and CUE respectively. Note that $\det(1-zG)$ will have the same distribution as if $G$ were a random unitary matrix chosen according to Haar measure. 

A purpose in this paper is to examine mixed moments of characteristic polynomials from the ACUE. In his blog post (see Remark 7), Tao made the remarkable observation that for quite large powers, moments of characteristic polynomials associated to $\ACUE(N)$ and $\CUE(N)$ agree:

\begin{theorem}[Tao]
\label{thm:Tao}
For positive integers $K, L \leq N$,
\begin{multline*}
\mathbb{E}_{\mathrm{ACUE}(N)} \Big[\prod_{k=1}^K \overline{\det(1-u_k g)} \prod_{\ell=1}^L \det(1-v_l g) \Big]\\
= \mathbb{E}_{\mathrm{CUE}(N)}\Big[ \prod_{k=1}^K \overline{\det(1-u_k G)} \prod_{\ell=1}^L \det(1-v_l G) \Big].
\end{multline*}
\end{theorem}

This allows one to compute a large range of moments for the ACUE using known results for the CUE. Nonetheless it is interesting to ask if a closed formula can be found that allows for the computation of all moments, and this is a main result of this paper.

In order to state it, for an integer $\ell$ and positive integer $m$, we introduce the notation $[\ell]_m \in \{0,1,...,m-1\}$ to be the reduction of $\ell$ modulo $m$, and define the function
\begin{equation}
\label{eq:H_def}
H_{N,\ell}(v):= \begin{cases} 0 & \textrm{if}\; 0 \leq [\ell]_{2N} \leq N-1 \\ v^{[\ell]_{2N} - N} & \textrm{if}\; N \leq [\ell]_{2N} \leq 2N-1. \end{cases}
\end{equation}

\begin{theorem}
\label{thm:main}
For $N, K, L \geq 1$, and $v_1,...,v_{K+L} \in \mathbb{C}$,
\begin{equation*}
\mathbb{E}_{\mathrm{ACUE}(N)} \Big[ \det(g)^{-K} \prod_{k=1}^{K+L} \det(1+ v_k g) \Big] = \frac{\det\big(\phi_i(v_j)\big)_{i,j=1}^{K+L}}{\Delta_{K+L}(v_1,...,v_{K+L})},
\end{equation*}
where
$$
\phi_i(v) = \phi_i^{(K,L;N)}(v) := \begin{cases} v^{N+L+K-i} - v^L \, H_{N,\, K-i}(v) & \textrm{for}\; 1 \leq i \leq K \\ v^{K+L-i} - v^{L+N-1} \, H_{N,\, i-K-1}(1/v) & \textrm{for}\; K+1 \leq i \leq K+L. \end{cases}
$$
\end{theorem}

Note that for associated to $\ACUE(N)$ (or indeed any unitary matrix $g$) we have 
\begin{equation}
\label{eq:func_eq}
\det(g)^{-1}\det(1+vg) = v^N \det(1+v^{-1} \overline{g}),
\end{equation}
so that this formula indeed allows for the computation of mixed moments of characteristic polynomials and their conjugates.

This should be compared to the analogous result for the CUE; we state this result in the formalism of Bump-Gamburd \cite{BuGa}.

\begin{theorem}[Prop. 4 of \cite{BuGa}]
\label{thm:BumpGamburd}
For $N, K, L \geq 1$,
\begin{align*}
\mathbb{E}_{\mathrm{CUE}(N)} \Big[\det(G)^{-K} \prod_{k=1}^{K+L} \det(1+ v_k G) \Big]&= s_{\langle N^k \rangle}(v_1,...,v_{K+L}) \\ 
&= \frac{\det\big(\psi_i(v_j)\big)_{i,j=1}^{K+L}}{\Delta_{K+L}(v_1,...,v_{K+L})},
\end{align*}
where
$$
\psi_i(v) = \psi_i^{(K,L;N)}(v) := \begin{cases} v^{N+L+K-i} & \textrm{for}\; 1 \leq i \leq K \\ v^{K+L-i} & \textrm{for}\; K+1 \leq i \leq K+L \end{cases}.
$$
\end{theorem}
The determinantal ratio here is just a definition of the Schur polynomial $s_{\langle N^K \rangle}(v_1,...,v_{K+L})$ associated to the partition $\langle N^K \rangle = (N,...,N),$ with $K$ parts.

An example makes the pattern of the matrices in the numerators of the right hand sides of Theorems \ref{thm:main} and \ref{thm:BumpGamburd} easier to see. If $K= 5,\, L=4,\, N=2$, columns in the variable $v$ will be:
\begin{align*}
\textrm{For ACUE:}\; \begin{pmatrix} \phi_1(v) \\ \phi_2(v) \\ \phi_3(v) \\ \phi_4(v) \\ \phi_5(v) \\ -- \\ \phi_6(v)  \\ \phi_7(v) \\ \phi_8(v) \\ \phi_9(v) \end{pmatrix}
= \begin{pmatrix} v^{10} \\ v^9 - v^5 \\ v^8- v^4 \\ v^7 \\ v^6 \\ -- \\ v^3 \\ v^2 \\ v - v^5 \\ 1 - v^4 \end{pmatrix}, \hspace{30pt}
\textrm{For CUE:}\; \begin{pmatrix} \psi_1(v) \\ \psi_2(v) \\ \psi_3(v) \\ \psi_4(v) \\ \psi_5(v) \\ -- \\ \psi_6(v)  \\ \psi_7(v) \\ \psi_8(v) \\ \psi_9(v) \end{pmatrix}
= \begin{pmatrix} v^{10} \\ v^9 \\ v^8 \\ v^7 \\ v^6 \\ -- \\ v^3 \\ v^2 \\ v  \\ 1  \end{pmatrix},
\end{align*}
where the line serves only to visually separate the block with indices $i \leq K$ from the block with indices $i \geq K+1$.

Note that if $K, L \leq N$ then we have
\begin{align*}
&K-i \leq N-1 \quad &&\textrm{for}\; 1 \leq i \leq K \\
&i-K-1 \leq N-1 \quad &&\textrm{for}\; K+1 \leq i \leq K+L,
\end{align*} 
so it follows by examining the definition of $H_{N,\ell}$ that $\phi^{(K,L; N)} = \psi^{(K,L; N)}$ in the above determinantal formulas. Thus these formulas recover the observation of Tao in Theorem \ref{thm:Tao}. By contrast if $K > N$ or $L > N$ these formulas show the moments for these models differ, despite having a closely related structure.

In fact it is by specializing the following formula for averages of ratios of characteristic polynomials that we derive Theorem \ref{thm:main}. 

\begin{theorem}
	\label{thm:ratios}
	For $N$ and $J$ positive integers, and $v_1,...,v_J$ complex numbers and $u_1,...,u_J$ complex numbers which are not $2N$-th roots of unity with $v_i \neq u_j$ for all $i,j$,
	\begin{equation}
		\label{eq:ratios}
		\EEACUE \Bigg[\frac{\prod_{j=1}^J \det(1+v_j g)}{\prod_{j=1}^J \det(1+u_j g)} \Bigg]
		= \frac{1}{\det\Big(\frac{1}{u_i-v_j}\Big)} \det\left(\frac{1}{u_i-v_j} \mathfrak{e}_N(u_i,v_j)\right),
	\end{equation}
	where the determinants on the right hand side are of $J \times J$ matrices, over the indices $1 \leq i,j \leq J$, and
	$$
	\mathfrak{e}_N(u,v) := \frac{1-u^N v^N}{1-u^{2N}}.
	$$
\end{theorem}

This formula in turn is a consequence of a general formula introduced by Borodin-Olshanski-Strahov in \cite{BoOlSt} for computing the average of ratios of characteristic polynomials associated to what they call Giambelli-compatible point processes. We will show the ACUE falls into this class of point processes and then specialize their result; see Theorem \ref{thm:BOS} below.

Theorem \ref{thm:ratios} may be compared to an analogous formula for the CUE (see e.g. \cite[Thm. 4.2]{Ro}, \cite[Thm. 5.4]{ChNaNi}, or \cite[(4.35)]{KiGu}):

\begin{theorem}
	\label{thm:ratios_CUE}
	For $N$ and $J$ positive integers, and $v_1,...,v_J$ complex numbers and $u_1,...,u_J$ complex numbers which do not lie on the unit circle with $v_i \neq u_j$ for all $i,j$,
	\begin{equation*}
		\EECUE \Bigg[\frac{\prod_{j=1}^J \det(1+v_j G)}{\prod_{j=1}^J \det(1+u_j G)} \Bigg]
		= \frac{1}{\det\Big(\frac{1}{u_i-v_j}\Big)} \det\left(\frac{1}{u_i-v_j} e_N(u_i,v_j)\right),
	\end{equation*}
	where the determinants on the right hand side are of $J \times J$ matrices, over the indices $1 \leq i,j \leq J$, and
	$$
	e_N(u,v) := \begin{cases} 1 & \textrm{if}\; |u|< 1, \\ v^N/u^N & \textrm{if}\; |u| > 1.\end{cases}
	$$
\end{theorem}

From Theorem \ref{thm:ratios}, a possible strategy for proving Theorem \ref{thm:main} is evident: we take appropriately scaled limits, with each $u_i$ tending either to $0$ or $\infty$ in order to recover the average appearing in Theorem \ref{thm:main}. Doing so nonetheless involves several nontrivial determinantal manipulations.

There is at least one alternative strategy for proving Theorems \ref{thm:main} and \ref{thm:ratios}, and this is to rely on the theory of orthogonal polynomials. This method has been used to derive similar formulas for moments and averages of ratios of characteristic polynomials in several random matrix ensembles; see for instance \cite{BrHi,BaDeSt,JoKeMe} for moments and \cite{StFy, BoSt} for ratios. One difficulty in the orthogonal polynomial method is that the finitely supported weights which define the ACUE allow for at most a finite collection of monic orthogonal polynomials. It would be interesting to see if this difficulty can be overcome to give alterative proofs of Theorems \ref{thm:main} or \ref{thm:ratios}.

It is perhaps a little surprising that moments of characteristic polynomials from the ACUE have a structure related to those from the CUE even for very large powers. This may ultimately be seen as a consequence of the similarity between Theorems \ref{thm:ratios} and \ref{thm:ratios_CUE} for ratios; another purpose of this paper is to provide an explanation of how ratio formulas like Theorem \ref{thm:ratios} can be used to derive moment formulas like Theorem \ref{thm:main}. It will be evident that the same method could be used to deduce Theorem \ref{thm:BumpGamburd} from Theorem \ref{thm:ratios_CUE} as well.

We note that formulas for the averages of ratios of characteristic polynomials in the CUE usually are written in a form involving a sum over `swaps', involving a slightly different formalism than Theorem \ref{thm:ratios_CUE},  -- see for instance \cite[Prop 2.1]{CoFoSn}, \cite[Cor. 1.2]{CoFaZi}, or \cite[Thm. 3]{BuGa}. By use of the functional equation, these formulas can be deduced from Theorem \ref{thm:ratios_CUE}. For instance, the $J=2$ case of Theorem \ref{thm:ratios_CUE} entails the following: for complex numbers $\alpha, \beta, \gamma, \delta$ with $|\gamma|, |\delta| < 1$,
\begin{multline*}
\EECUE \Bigg[ \frac{\det(1-\alpha \,G) \det(1-\beta\, \overline{G})}{\det(1-\gamma\, G) \det(1-\delta\, \overline{G})} \Bigg] =  \frac{\beta^N}{\delta^N} \EECUE \Bigg[ \frac{\det(1-\alpha \,G) \det(1-\beta^{-1}\, G)}{\det(1-\gamma\, G) \det(1-\delta^{-1}\, G)} \Bigg]\\
= \frac{(1-\beta\gamma) (1-\alpha \delta)}{(1-\delta \gamma) (1-\alpha \beta)} + (\alpha\beta)^N \frac{(1-\gamma \alpha^{-1}) (1-\delta \beta^{-1})}{(1-\alpha^{-1} \beta^{-1})(1-\gamma \delta)}.
\end{multline*}
Note that this formula is valid only for $|\gamma|, |\delta| < 1.$ If instead for instance $|\gamma| < 1$ and $|\delta| > 1$, the left hand side would just work out to $1$.

By using the functional equation \eqref{eq:func_eq} for $\det(1+v g)$ one can derive expressions of this sort for the ACUE as well. For instance, for complex numbers $\alpha, \beta, \gamma, \delta$ with neither $\gamma$ nor $\delta$ equal to $2N$-th roots of unity, Theorem \ref{thm:ratios} reveals,
\begin{multline*}
\EEACUE \Bigg[ \frac{\det(1-\alpha \,g) \det(1-\beta\, \overline{g})}{\det(1-\gamma\, g) \det(1-\delta\, \overline{g})} \Bigg] = \frac{\beta^N}{\delta^N} \EEACUE \Bigg[ \frac{\det(1-\alpha \,g) \det(1-\beta^{-1}\, g)}{\det(1-\gamma\, g) \det(1-\delta^{-1}\, g)} \Bigg] \\
= \frac{(1-\beta\gamma) (1-\alpha \delta)}{(1-\delta \gamma) (1-\alpha \beta)} \Big( \frac{1- \alpha^N \gamma^N}{1-\gamma^{2N}}\Big)\Big(\frac{1-\beta^N \delta^N}{1-\delta^{2N}}\Big)\\
 + (\alpha\beta)^N \frac{(1-\gamma \alpha^{-1}) (1-\delta \beta^{-1})}{(1-\alpha^{-1} \beta^{-1})(1-\gamma \delta)} \Big( \frac{1- \beta^{-N} \gamma^N}{1-\gamma^{2N}}\Big)\Big(\frac{1-\alpha^{-N} \delta^N}{1-\delta^{2N}}\Big).
\end{multline*}
Note that in this case there is no need to assume that $|\gamma|, |\delta| < 1$. Indeed, the right and left hand sides are meromorphic in the variables $\gamma$ and $\delta$, with singularities only at $2N$-th roots of unity.

This procedure can be used to obtain formulas for $J > 2$ as well. But for mixed ratios of more than two characteristic polynomials, expansions like this for the ACUE seem to become increasingly more complicated than those for the CUE; by contrast the determinantal formula of Theorem \ref{thm:ratios} remains relatively simple for all $J$.

It is natural to ask whether Theorems \ref{thm:main} or \ref{thm:BOS} shed light on any number theoretic phenomena. A typical question in number theory involves moments of the Riemann zeta-function in which powers $K$ and $L$ are fixed or grow slowly. Theorem \ref{thm:Tao} of Tao is certainly of interest in this regard, but because $K$ and $L$ must be of size at least $N$ before Theorem \ref{thm:main} sees a difference between the CUE and ACUE prediction, it does not seem that the new information in this theorem will shed light on these sorts of questions. On the other hand, uniform estimates for moments can be of some interest in determining extreme values of L-functions (see e.g. \cite[Sec. 7]{Sou}), and Theorem \ref{thm:main} may be of some use in examining alternative possibilities here. Furthermore Theorem \ref{thm:ratios} suggests a hypothetical `alternative ratio formula' for the Riemann zeta-function -- a formula which one would like to rule out but cannot at present. This is discussed further in Section \ref{sec:zetafunction}.

\vspace{10pt}

\noindent \textbf{Acknowledgements:} We thank David Farmer and Ofir Gorodetsky for very useful references, comments, and corrections. We are also grateful to the anonymous referee for a careful reading and useful comments and corrections. B.R. received partial support from an NSERC grant and US NSF FRG grant 1854398.

\section{The ratio formula: Theorem \ref{thm:ratios}}

In this section we prove Theorem \ref{thm:ratios}. Our starting point is an application of a general formula of Borodin-Olshanski-Strahov to the ACUE.

\begin{theorem}[A Borodin-Olshanski-Strahov Formula for ACUE]
	\label{thm:BOS}
	For $N$ and $J$ positive integers, $v_1,...,v_J$ complex numbers, and $u_1,...,u_J$ complex numbers which are not $2N$-th roots of unity with $v_i \neq u_j$ for all $i,j$,
	\begin{multline}
		\label{eq:BOS}
		\EEACUE \Bigg[\frac{\prod_{j=1}^J \det(1+v_j g)}{\prod_{j=1}^J \det(1+u_j g)} \Bigg] \\
		= \frac{1}{\det\Big(\frac{1}{u_i-v_j}\Big)} \det\left(\frac{1}{u_i-v_j} \EEACUE \Big[\frac{\det(1+v_j g)}{\det(1+u_i g)}\Big]\right),
	\end{multline}
	where the determinants on the right hand side are of $J \times J$ matrices, over the indices $1 \leq i,j \leq J$.
\end{theorem}

\begin{proof}
This requires only minor modifications of formulas in \cite{BoOlSt}. Claims I and II of that paper show that if $\mu$ is a measure on $\mathbb{C}$ with finite moments and if a point process consisting of $N$ points $\{z_1,...,z_N\}$ in $\mathbb{C}$ has a joint density given by
$$
(\mathrm{const.}) \prod_{1\leq i < j \leq N} |z_i-z_j|^2 \prod_{i=1}^N \mu(dz_i),
$$
then as a formal powers series
\begin{multline*}
\mathbb{E} \big[H(\alpha_1)\cdots H(\alpha_J) E(\beta_1)\cdots E(\beta_J) \big]\\
= \frac{1}{\det\Big(\frac{1}{\alpha_i+\beta_j}\Big)} \det\left(\frac{1}{\alpha_i+\beta_j} \mathbb{E} \big[H(\alpha_i) E(\beta_j)\big]\right),
\end{multline*}
where
$$
H(\alpha) := \frac{1}{\prod_{j=1}^N (1-z_j \alpha^{-1})}, \quad E(\beta):= \prod_{j=1}^N(1+z_j \beta^{-1}).
$$

This is only claimed for a measure $\mu$ supported on $\mathbb{R}$, but the proof applies with no change to measures supported on $\mathbb{C}$, except that in the proof of Theorem 3.1 the moments $A_n = \int_{\mathbb{R}} x^n \mu(dx)$ must be replaced by $A_{n,m} = \int_\mathbb{C} z^n \overline{z}^m\, \mu(dz)$ and later in the proof $A_{\lambda_i + N-i + N-j}$ must be replaced by $A_{\lambda_i + N - i\,,\, N-j}$.

The point process ACUE is induced by such a joint density where $\mu$ is a probability measure uniform on the $2N$-th roots of unity in $\mathbb{C}$. This identity may be seen to be true not just for formal powers series but for functions $H(\alpha)$, $E(\beta)$ by considering the case $|\alpha_1|,...,|\alpha_J| > 1$ (where all power series will converge absolutely) and then meromorphically continuing to all $\alpha_1,...,\alpha_J$.

Finally, we arrive at \eqref{eq:BOS} simply by setting $\alpha_j = - u_j^{-1}$, $\beta_j = v_j^{-1}$ and simplifying the resulting determinants.
\end{proof}

The remainder of this section is therefore devoted to understanding the expectation which occurs on the right hand side of \eqref{eq:BOS}, accomplished in Proposition \ref{prop:1by1ratio} below.

\begin{lemma}
\label{lem:hookschur}
Consider a hook partition $(a,1^b)$ with $a\geq 1$ and $b \geq 0$ of length $b+1 \leq N$. For the Schur polynomial $s_{(a,1^b)}$ associated to this partition in the variables $e(\vartheta_1),...,e(\vartheta_N)$ of the $\ACUE(N)$, we have
$$
\EEACUE\, s_{(a,1^b)}  = \begin{cases} (-1)^b & \textrm{if}\quad a+b \equiv 0 \Mod {2N} \\ 0 & \textrm{otherwise}. \end{cases}
$$
\end{lemma}

\begin{proof}
Label $\omega_j = e(\vartheta_j)$ so that for a partition $\lambda$ of length $\ell(\lambda) \leq N$,
$$
s_\lambda = \frac{\det(\omega_i^{\lambda_j + N -j})}{\det(\omega_i^{N-j})},
$$
where if $\ell(\lambda) < N$ we adopt the convention $\lambda_{\ell(\lambda)+1} = \cdots = \lambda_N = 0$, and the determinants above are $N \times N$.

Note that $\det(\omega_i^{N-j}) = \Delta(\omega_1,...,\omega_N)$. Hence from the definition \eqref{eq:ACUEdef} of the ACUE,
\begin{multline}
\label{eq:detsum}
\EEACUE\, s_\lambda \\ = \frac{1}{N! \, (2N)^N} \sum_{t_1,...,t_N} \det\big(e\big((\lambda_j + N - j) t_i\big)\big)  \overline{\det\big(e\big((N - j) t_i\big)\big)},
\end{multline}
where each index $t_i$ is summed over the set $\{0, \tfrac{1}{2N},...,\tfrac{2N-1}{2N}\}$. Expanding each determinant into a sum over permutations mapping $\{1,...,N\}$ to $\{1,...,N\}$ one sees
\begin{multline*}
\det\big(e\big(\lambda_j + N - j) t_i\big)\big)  \overline{\det\big(e\big(N - j) t_i\big)\big)} \\
= \sum_{\sigma, \pi \in S_N} (-1)^\sigma (-1)^\pi \prod_{i=1}^N e\big((\lambda_{\sigma(i)} + N - \sigma(i)) t_i - (N - \pi(i)) t_i\big).
\end{multline*}
Thus \eqref{eq:detsum} is
\begin{align}
\label{eq:permutationsum}
\notag & = \frac{1}{N!} \sum_{\sigma,\pi \in S_N} (-1)^\sigma (-1)^\pi \mathbf{1}\big[(\lambda_{\sigma(i)} - \sigma(i)) + \pi(i) \equiv 0 \Mod{2N}\quad \textrm{for all}\; i\big] \\
&= \sum_{\pi \in S_N} (-1)^\pi \mathbf{1}\big[\lambda_{i} - i + \pi(i) \equiv 0 \Mod{2N}\quad \textrm{for all}\; i\big],
\end{align}
where $\mathbf{1}[\,\cdot\,]$ denotes an indicator function, taking the value $1$ or $0$ depending on whether the proposition inside is true or false.

In the special case that $\lambda = (a,1^b)$ this sum has a simple evaluation. In that case any nonvanishing summand will have $\pi$ satisfying
\begin{align*}
a-1 + \pi(1) &\equiv 0 \Mod {2N} \\
&\textrm{and}\\
1 - 2 + \pi(2) &\equiv 0 \Mod {2N} \\
&\vdots \\
1 - (b+1) + \pi(b+1) &\equiv 0 \Mod {2N}\\
&\textrm{and}\\
0 - (b+2) + \pi(b+2) &\equiv 0 \Mod {2N} \\
&\vdots \\
0 - N + \pi(N) &\equiv 0 \Mod {2N}.
\end{align*}
Since $1 \leq \pi(i) \leq N$, the last $N-1$ of these equations force
$$
\pi(b+2) = b+2, \; ...., \; \pi(N) = N,
$$
$$
\pi(2) = 1, \; ...., \; \pi(b+1) = b.
$$
This forces 
$$
\pi(1) = b+1,
$$
and so at most one permutation $\pi$ makes a nonzero contribution to \eqref{eq:permutationsum}, and that contribution is nonzero if and only if $a+b \equiv 0 \Mod{2N}$, since $a + b = a - 1 + \pi(1)$. Since in cycle notation this permutation is $\pi = (b+1,\, b, \, ..., \, 2, \, 1)$ we have $(-1)^\pi = (-1)^b$ and this verifies the lemma.
\end{proof}

\begin{proposition}
\label{prop:1by1ratio}
For $v$ any complex number and $u$ any complex number which is not a $2N$-th root of unity,
$$
\EEACUE \frac{\det(1+vg)}{\det(1+ug)} = \frac{1- u^N v^N}{1-u^{2N}}.
$$
\end{proposition}

\begin{proof}
We first consider $|u|<1$. From a series expansion we have
\begin{equation}
\label{eq:ratio_expansion}
\frac{\det(1+vg)}{\det(1+ug)} = \sum_{j=0}^N \sum_{k=0}^\infty (-1)^k e_j h_k v^j u^k,
\end{equation}
where $e_j$ and $h_k$ are respectively elementary symmetric polynomials of degree $j$ and homogeneous symmetric polynomials of degree $k$ in the variables $e(\vartheta_1),...,e(\vartheta_N)$ associated to $\ACUE(N)$. Note that
$$
e_0 = h_0 = 1,
$$
while other terms can be expression in terms of Schur polynomials in the variables $e(\vartheta_1),...,e(\vartheta_N)$:
\begin{align*}
e_j h_0 &= s_{(1^j)} \quad \textrm{for}\; 1 \leq j \leq N,\\
e_0 h_k &= s_{(k)}  \quad \textrm{for}\; k \geq 1,\\
e_j h_k &= s_{(k+1,1^{j-1})} + s_{(k,1^j)}  \quad \textrm{for}\; 1 \leq j \leq N-1, \, k\geq 1,\\
e_N h_k &= s_{(k+1,1^{N-1})}  \quad \textrm{for}\; k\geq 1,
\end{align*}
with the first two identities following from the combinatorial definition of Schur functions \cite[Sec. 7.10]{St}, and the last two from the Pieri rule \cite[Thm. 7.15.7]{St}. From Lemma \ref{lem:hookschur} it thus follows
$$
\EEACUE e_j h_k = \begin{cases} 
	1 & \textrm{if}\; j = 0, \; k \equiv 0 \Mod{2N}, \\ 
	(-1)^{N-1} & \textrm{if}\; j = N, \; k \equiv N \Mod{2N}, \\
	0 & \textrm{otherwise}.
	\end{cases}
$$
Hence from \eqref{eq:ratio_expansion},
\begin{align*}
\EEACUE \frac{\det(1+vg)}{\det(1+ug)} & = (1+ u^{2N} + u^{4N} + \cdots) - (v^N u^N + v^N u^{3N} + v^N u^{5N} + \cdots) \\
&= \frac{1-u^N v^N}{1-u^{2N}},
\end{align*}
for $|u| < 1$ and all $v$. The result then follows by analytic continuation.
\end{proof}

Thus we have:

\begin{proof}[Proof of Theorem \ref{thm:ratios}]
Apply Proposition \ref{prop:1by1ratio} to Theorem \ref{thm:BOS}.
\end{proof}

\section{The moment formula: Theorem \ref{thm:main}}

Our technique in proving Theorem \ref{thm:main} will be to condense the determinants in \eqref{eq:ratios} by letting all $u_i\rightarrow 0$. We begin with several lemmas that are useful for that purpose.

The following is a slight generalization of Lemma 1 of \cite{Me}.

\begin{lemma}[Determinantal Condensation Identity]
\label{lem:condense}
Take $q \leq J$. For $f_1, f_2,...,f_J$ functions (mapping $\mathbb{R}$ to $\mathbb{C}$) that are at least $q$ times continuously differentiable at the point $a$,
\begin{equation}
\label{eq:condense}
\lim_{u_1,...,u_q \rightarrow a} \frac{1}{\Delta(u_q,...,u_1)}\det\big(f_j(u_i)\big)_{i,j=1}^J = \det \begin{pmatrix} \Big(\frac{1}{(i-1)!} f_j^{(i-1)}(a)\Big)_{i \leq q,\;\, j \leq J} \\ \big(f_j(u_i)\big)_{q+1 \leq i \leq J,\;\, j \leq J}\end{pmatrix},
\end{equation}
where on the left hand side the limit is taken in the order that first $u_1 \rightarrow u_2$, then $u_2 \rightarrow u_3$, ... , $u_{q-1}\rightarrow u_q$, and finally $u_q \rightarrow a$.
\end{lemma}

\begin{proof}
We prove this identity by induction, viewing $\Delta(u_1) = 1$ for the $q = 1$ case, which then becomes trivial. Suppose then that \eqref{eq:condense} has been proved for a limit in $q-1$ variables. This implies for a limit in $q$ variables,
\begin{align*}
&\lim_{u_1,...,u_q \rightarrow a} \frac{1}{\Delta(u_q,...,u_1)}\det\big(f_j(u_i)\big)_{i,j=1}^J \\
&= \lim_{u_q\rightarrow a} \;\lim_{u_{q-1}\rightarrow u_q} \frac{1}{(u_q-q_{q-1})^{q-1}} \det \begin{pmatrix} \Big(\frac{1}{(i-1)!} f_j^{(i-1)}(u_{q-1})\Big)_{i \leq q-1, j \leq J} \\ \big(f_j(u_i)\big)_{i \geq q, j \leq J}\end{pmatrix}.
\end{align*}
But Taylor expanding the entries of row $q$ as
$$
f_j(u_q) = \sum_{i=1}^q \frac{f_j^{(i-1)}(u_{q-1})}{(i-1)!}  (u_q-u_{q-1})^{i-1} + O((u_q-u_{q-1})^{q}),
$$
and using multilinearity of the determinant to cancel out the first $q-1$ terms of the above sum in row $q$, the claimed result quickly follows.
\end{proof}

\begin{remark}
It is likely that a result of this sort remains true no matter the path along which a limit is taken (perhaps with further analytic conditions on the functions $f_j$), but we won't require that in what follows.
\end{remark}

\begin{remark} 
It is easy to see by permuting rows of the determinant that this result also implies
\begin{equation}
\label{eq:reverseorder_condense}
\lim_{u_1,...,u_q \rightarrow a} \frac{1}{\Delta(u_1,...,u_q)}\det\big(f_j(u_i)\big)_{i,j=1}^J = \det \begin{pmatrix} \Big(\frac{1}{(q-i)!} f_j^{(q-i)}(a)\Big)_{i \leq q, j \leq J} \\ \big(f_j(u_i)\big)_{i \geq q+1, j \leq J}\end{pmatrix}
\end{equation}
and
\begin{equation}
\label{eq:second_condense}
\lim_{u_{q+1},...,u_J \rightarrow a} \frac{1}{\Delta(u_J,...,u_{q+1})}\det\big(f_j(u_i)\big)_{i,j=1}^J = \det \begin{pmatrix} \big(f_j(u_i)\big)_{i \leq q, j \leq J} \\ \Big(\frac{1}{(i - q-1)!} f_j^{(i - q-1)}(a)\Big)_{i \geq q+1, j \leq J}\end{pmatrix},
\end{equation}
where in this last equation the limit is taken in the order $u_{q+1}\rightarrow u_{q+2}$,..., $u_{J-1}\rightarrow u_J$, $u_J \rightarrow a$.
\end{remark}

In applying this lemma we need the following computation.

\begin{lemma}
\label{lem:derivatives}
For integers $\ell \geq 0$ and $N\geq 1$,
$$
\lim_{u\rightarrow 0} \frac{1}{\ell!} \frac{d^\ell}{du^\ell} \Big(\frac{1}{u-v} \frac{1-u^N v^N}{1-u^{2N}}\Big) = -p_{N,\ell}(v),
$$
for $p_{N,\ell}$ defined by
\begin{equation}
\label{eq:p_def}
p_{N,\ell}(v) := \frac{1}{v^{\ell+1}} - v^{N-1} \, H_{N,\ell}(1/v) = \frac{1}{v^{\ell+1}} - \begin{cases} 0 & \textrm{if}\; 0 \leq [\ell]_{2N} \leq N-1 \\ v^{2N-1-[\ell]_{2N}} & \textrm{if}\; N \leq [\ell]_{2N} \leq 2N-1.\end{cases}
\end{equation}
\end{lemma}

\begin{proof}
Note that we have
\begin{multline*}
\frac{1}{u-v} \frac{1-u^N v^N}{1-u^{2N}} = - \frac{1}{v} \frac{1}{1-u/v} + \frac{u^N-v^N}{u-v} \frac{u^N}{1-u^{2N}} \\
= -\Big(\frac{1}{v} + \frac{u}{v^2} + \frac{u^2}{v^3}+\cdots \Big) + (v^{N-1} + v^{N-2} u + \cdots + u^{N-1})(u^N + u^{3N} + \cdots),
\end{multline*}
taking a series expansion around $u=0$. Since the quantity on the left hand side of the Lemma is exactly the coefficient of $u^\ell$ in this expansion, the claim follows by inspection.
\end{proof}

\begin{lemma}(Cauchy Determinant Formula)
\label{lem:Cauchy}
For $u_1,...,u_J$ and $v_1,...,v_J$ collections of complex numbers with no elements in common,
$$
\det\Big(\frac{1}{u_i - v_j}\Big)_{i,j=1}^J = \frac{\Delta(u_J,...,u_1) \Delta(v_1,...,v_J)}{\square(u;v)}
$$
where
$$
\square(u;v):= \prod_{i=1}^J \prod_{j=1}^J (u_i-v_j).
$$
\end{lemma}

\begin{proof}
See for instance \cite[Part 7, \S 1, Ex. 3]{PoSz}.
\end{proof}

We can now give a proof of the moment formula for ACUE.

\begin{proof}[Proof of Theorem \ref{thm:main}]
We set $u = (u_1,...,u_K)$ and $u' = (u'_1,...,u'_L)$ as abbreviations for ordered lists, and let $u \cup u':= (u_1,...,u_K, u'_1,...,u'_L)$ be an (ordered) concatenation of these lists. We abbreviate $\Delta(u) = \Delta(u_1,...,u_K)$ and also use the notation $\widetilde{\Delta}(u) = \Delta(u_K,...,u_1) = (-1)^{K(K-1)/2}\Delta(u)$.

Our starting point is the identity
\begin{multline}
\label{eq:tolimit}
\EEACUE\Big[\det(g)^{-K} \prod_{k=1}^{K+L}\det(1+v_k g)\Big] \\
= \lim_{u \rightarrow \infty} u_1^N \cdots u_K^N \lim_{u' \rightarrow 0} E_N(u \cup u'\, ; \, v),
\end{multline}
where we define
$$
E_N(u \cup u'\, ; \, v) := \EEACUE\Bigg[\frac{\prod_{k=1}^{K+L} \det(1+v_k g)}{\prod_{k=1}^K \det(1+u_k g) \prod_{\ell=1}^L \det(1+ u_\ell' g)}\Bigg].
$$
The limits $u \rightarrow \infty$ and $u' \rightarrow 0$ mean $u_1,...,u_K \rightarrow \infty$ and $u'_1,...,u'_L \rightarrow 0$. In what follows we will take these in the order $u'_1 \rightarrow u'_2$, ..., $u'_{L-1}\rightarrow u'_L$, $u'_L\rightarrow 0$ and $u_1 \rightarrow u_2$, ..., $u_{K-1}\rightarrow u_K$, $u_K\rightarrow \infty$ so that Lemma \ref{lem:condense} can easily be applied.

For notational reasons we write
$$
f_N(u,v):= \frac{1}{u-v}\frac{1-u^N v^N}{1-u^{2N}}.
$$
We use Theorem \ref{thm:ratios} and Lemma \ref{lem:Cauchy} to see,
\begin{align*}
E_N(u \cup u'\, ; \, v) &= \frac{\square(u \cup u'\, ; \, v)}{\widetilde{\Delta}(u\cup u') \Delta(v)} \det\begin{pmatrix} \Big(f_N(u_i,v_j)\Big)_{i\leq K, j \leq K+L} \\ \Big(f_N(u'_{i-K},v_j)\Big)_{i\geq K+1, j \leq K+L} \end{pmatrix} \\
&= \frac{\square(u \, ; \, v) \square(u' \, ; \, v)}{\widetilde{\Delta}(u) \widetilde{\Delta}(u') \square(u';u) \Delta(v)} \det\begin{pmatrix} \Big(f_N(u_i,v_j)\Big)_{i\leq K, j \leq K+L} \\ \Big(f_N(u'_{i-K},v_j)\Big)_{i\geq K+1, j \leq K+L} \end{pmatrix}.
\end{align*}
Taking a limit $u'\rightarrow 0$ and using Lemma \ref{lem:condense} -- in particular its consequence \eqref{eq:second_condense} -- and Lemma \ref{lem:derivatives},
$$
\lim_{u'\rightarrow 0} E_N(u \cup u'\, ; \, v) = \frac{\square(u;v) \prod_{k=1}^{K+L}(-v_k)^L}{\widetilde{\Delta}(u) \prod_{k=1}^K (-u_k)^L \Delta(v)} \det\begin{pmatrix} \Big(f_N(u_i,v_j)\Big)_{i\leq K, j \leq K+L} \\ \Big(-p_{N,i-K-1}(v_j)\Big)_{i\geq K+1, j \leq K+L} \end{pmatrix}.
$$
But note the easily verified functional equation
$$
f_N(u,v) = -f_N(u^{-1}, v^{-1}) v^{N-1} u^{-(N+1)}.
$$
Thus
\begin{multline}
\label{eq:almostdone}
u_1^N\cdots u_K^N \lim_{u'\rightarrow 0} E_N(u \cup u'\, ; \, v) \\
= (-1)^L \frac{\square(u;v) \prod_{k=1}^{K+L} v_k^L}{\widetilde{\Delta}(u) \Delta(v)} \prod_{k=1}^K u_k^{-L-1} \cdot\det\begin{pmatrix} \Big(- v_j^{N-1}f_N(u_i^{-1},v_j^{-1})\Big)_{i\leq K, j \leq K+L} \\ \Big(-p_{N,i-K-1}(v_j)\Big)_{i\geq K+1, j \leq K+L} \end{pmatrix}.
\end{multline}

For fixed $v$, we have as $u\rightarrow\infty$,
$$
\frac{\square(u;v)}{\widetilde{\Delta}(u)} \prod_{k=1}^K u_k^{-L-1} = \frac{\square(u;v) \prod_{k=1}^K u_k^{-L-1}}{\prod_{k=1}^K u_k^{K-1} \Delta\big(\frac{1}{u_1},...,\frac{1}{u_K}\big)} \sim \frac{1}{\Delta\big(\frac{1}{u_1},...,\frac{1}{u_K}\big)}.
$$
Applying Lemma \eqref{lem:condense} -- with its consequence \eqref{eq:reverseorder_condense} this time -- and Lemma \ref{lem:derivatives}, the limit of \eqref{eq:almostdone} as $u \rightarrow \infty$ is
\begin{align*}
& = (-1)^L \frac{\prod_{k=1}^{K+L} v_k^L}{\Delta(v)} \det\begin{pmatrix} \Big( v_j^{N-1}p_{N,K-i}(v_j^{-1})\Big)_{i\leq K, j \leq K+L} \\ \Big(-p_{N,i-K-1}(v_j)\Big)_{i\geq K+1, j \leq K+L} \end{pmatrix} \\
& = \frac{1}{\Delta(v)} \det\begin{pmatrix} \Big( v_j^{N+L-1}p_{N,K-i}(v_j^{-1})\Big)_{i\leq K, j \leq K+L} \\ \Big( v_j^L p_{N,i-K-1}(v_j)\Big)_{i\geq K+1, j \leq K+L} \end{pmatrix}.
\end{align*}
By inspection of matrix entries, the above is
$$
= \frac{\det\big(\phi_i(v_j)\big)_{i,j=1}^{K+L}}{\Delta(v)}.
$$
Recalling \eqref{eq:tolimit}, this is exactly what we sought to prove.
\end{proof}

\section{Hypothetical implications for ratios of $\zeta(s)$}
\label{sec:zetafunction}

Let us briefly and somewhat informally discuss these results in the context of the distribution of the Riemann zeta-function. For the sake of this discussion, suppose the Riemann Hypothesis is true, and label the nontrivial zeros of the zeta-function by $\{1/2+i\gamma_j\}_{j\in \mathbb{Z}}$, so that $\gamma_j \in \mathbb{R}$ for all $j$. What is widely believed about the local distribution of zeros concerns two point processes, the first point process (associated to a large parameter $T$) given by
\begin{equation}
	\label{eq:zeta_points}
	\Big\{ \frac{\log T}{2\pi}(\gamma_j-t)\Big\}_{j \in \mathbb{Z}}
\end{equation}
where $t \in [T,2T]$ is chosen randomly and uniformly, and the second point process (associated to a large parameter $N$) given by
\begin{equation}
\label{eq:CUE_pointprocess}
\{N\theta_i\}_{i=1,...,N}
\end{equation}
where $\theta_1,...,\theta_N \in [-1/2,1/2)$ are identified with the points $e(\theta_1),...,e(\theta_N)$ of $\CUE(N)$. The widely believed \emph{GUE Hypothesis} states that as $T\rightarrow\infty$ and $N\rightarrow\infty$ both point processes \eqref{eq:zeta_points} and \eqref{eq:CUE_pointprocess} tend to the same limiting point process. (This means that randomly generated configurations of points from these two processes will look similar near the origin of the real line.) 

The name GUE Hypothesis has historical origins; GUE refers to the \emph{Gaussian Unitary Ensemble}, an ensemble of random matrices which, like CUE, locally tends to this same limiting point process, but which was investigated earlier. Much of what is known rigorously about the limiting distribution of the points in \eqref{eq:zeta_points} is due to Montgomery \cite{Mo}, Hejhal \cite{He}, and Rudnick and Sarnak \cite{RuSa}, who showed that the correlation functions of the points in \eqref{eq:zeta_points} agree with those of this limiting distribution up to a band-limited resolution. Numerical work, beginning with that of Odlyzko \cite{Od}, has given further support to the GUE Hypothesis.

The ACUE was first investigated as one alternative model of how zeros of the Riemann zeta-function might be spaced. In particular, one considers the point process (associated to a large parameter $N$) given by
\begin{equation}
\label{eq:ACUE_pointprocess}
\{N\vartheta_i + \tfrac{r}{2}\}_{i=1,...,N}
\end{equation}
where $\vartheta_1,...,\vartheta_N \in \{-\tfrac{1}{2}, \tfrac{-N+1}{2N},\tfrac{-N+2}{2N},..., \tfrac{N-1}{2N}\}$ are identified with the points $e(\vartheta_1),...,e(\vartheta_N)$ of $\ACUE(N)$, and $r \in [0,1)$ is chosen  independently, and uniformly at random. As $N\rightarrow\infty$ the point process \eqref{eq:ACUE_pointprocess} tends to a limiting process, called the AH point process in \cite{LaRo1}. The AH point process has correlation functions which mimic the limiting process for CUE in a fashion akin to what is known rigorously about zeta zeros from the results of Montgomery, Hejhal, and Rudnick and Sarnak (see \cite{LaRo2} for further discussion), but it also has gaps between points which are always half-integers. In this way it is one possible -- though likely not a unique -- candidate for a limiting distribution of the zeta-function point process \eqref{eq:zeta_points} which is compatible with what is currently known about the local distribution of zeros of the zeta-function and also with the so-called Alternative Hypothesis, a (widely disbelieved) conjecture that gaps between zeros always occur close to half-integer multiples of the mean spacing. 

For this reason \cite{Ta} gave the name \emph{AGUE (Alternative Gaussian Unitary Ensemble) Hypothesis} to the hypothetical claim that as $T\rightarrow \infty$ the zeta zero point process \eqref{eq:zeta_points} tends to the AH point process. As one would like to rule out the Alternative Hypothesis, one would like to rule out the stronger AGUE Hypothesis. 

More details on the AH point process can be found in the references \cite{LaRo1, Ta}, while further information on the Alternative Hypothesis in general can be found in \cite{Ba}.

A major impetus for studying mixed moments of characteristic polynomials $\det(1+u G)$ for the CUE came from the work of Keating-Snaith, who used information about CUE moments to make a conjecture regarding moments of the Riemann zeta-function \cite[Eq. (19)]{KeSn}. As first observed by Tao and as discussed in the introduction, the consequence of Theorem \ref{thm:Tao} that for sufficiently large $N$ mixed moments in the CUE and ACUE agree suggests that even should the zeros of the Riemann zeta-function be spaced according to the pattern of the ACUE, this could still be consistent with the Keating-Snaith moment conjecture.

The local spacing of zeros of the Riemann zeta-function is also closely related to the averages of ratios of shifts of the Riemann zeta function near the critical line. This perspective was first pursued by Farmer \cite{Fa1, Fa2} and has subsequently been investigated by others \cite{CoFaZi,CoSn1, CoSn2}. In particular note that from Theorem \ref{thm:ratios_CUE},
\begin{equation*}
\lim_{N\rightarrow\infty}\EECUE \Bigg[\prod_{j=1}^J\frac{\det(1 - e^{-\nu_j/N} G)}{\det(1 - e^{-\mu_j/N} G)} \Bigg]
= \frac{1}{\det\Big(\frac{1}{\nu_j - \mu_i}\Big)} \det\left(\frac{1}{\nu_j-\mu_i} e(\mu_i,\nu_j)\right),
\end{equation*}
for $\Re\, \mu_j \neq 0$ for all $j$, where
$$
e(\mu,\nu) := \begin{cases} 1 & \textrm{if}\; \Re\, \mu > 0 \\ e^{\mu-\nu} & \textrm{if}\; \Re\, \mu < 0.\end{cases}
$$

From the results proved in \cite{CoSn2,Ro} it can be seen that the claim
\begin{equation}
\label{eq:zetaratios}
\lim_{T\rightarrow\infty}\frac{1}{T}\int_T^{2T} \prod_{j=1}^J \frac{\zeta(1/2 + \nu_j/\log T + it)}{\zeta(1/2 + \mu_j/\log T + it)} \, dt = \frac{1}{\det\Big(\frac{1}{\nu_j - \mu_i}\Big)} \det\left(\frac{1}{\nu_j-\mu_i} e(\mu_i,\nu_j)\right)
\end{equation}
for $\Re\, \mu_j \neq 0$ for all $j$, is equivalent to the GUE Hypothesis. (In fact \cite{Ro} treats only real $\mu_j, \nu_j$, but the method can be adapted to complex values. There is a notational difference in \cite{Ro}; the function $E$ used there satisfies $E(\nu,\mu) = e(\mu,\nu)$ for the function $e$ used here.)

A belief in the AGUE Hypothesis would suggest that we replace characteristic polynomials $\det(1-u G)$ as they appear above by $\det(1- u e^{i 2\pi r/2N} g)$, where $r\in[0,1)$ is independent of $g$ and uniformly chosen.
For the ACUE, from Theorem \ref{thm:ratios} we have
\begin{equation*}
\lim_{N\rightarrow\infty} \EEACUE \Bigg[\prod_{j=1}^J\frac{\det(1 - e^{-\nu_j/N} g)}{\det(1 - e^{-\mu_j/N} g)} \Bigg]
= \frac{1}{\det\Big(\frac{1}{\nu_j - \mu_i}\Big)} \det\left(\frac{1}{\nu_j-\mu_i} \mathfrak{e}(\mu_i,\nu_j)\right),
\end{equation*}
for $\mu_j \notin \frac{i}{2}\mathbb{Z}$ for all $j$, where
$$
\mathfrak{e}(\mu,\nu) := \frac{1 - e^{-\mu-\nu}}{1 - e^{-2\mu}}.
$$
Hence on the assumption of the AGUE Hypothesis, one should instead expect for $\Re \, \mu_j \neq 0$ for all $j$,
\begin{multline}
\label{eq:alternative_zetaratios}
\lim_{T\rightarrow\infty}\frac{1}{T}\int_T^{2T} \prod_{j=1}^J \frac{\zeta(1/2 + \nu_j/\log T + it)}{\zeta(1/2 + \mu_j/\log T + it)}\, dt \\
= \lim_{N\rightarrow\infty} \int_0^1 \EEACUE \Bigg[\prod_{j=1}^J\frac{\det(1 - e^{-\nu_j/N}e^{i\pi r/N} g)}{\det(1 - e^{-\mu_j/N} e^{i\pi r/N} g)} \Bigg]\, dr\\
=  \int_0^1 \frac{1}{\det\Big(\frac{1}{\nu_j - \mu_i}\Big)} \det\left(\frac{1}{\nu_j-\mu_i} \mathfrak{e}(\mu_i-i\pi r,\nu_j-i\pi r)\right) dr \\
= \frac{1}{\det\Big(\frac{1}{\nu_j - \mu_i}\Big)}\int_{|z|=1} \det\left(\frac{1}{\nu_j-\mu_i} \frac{1 - z e^{-\mu_i-\nu_j}}{1 - z e^{-2\mu_i}}\right) \frac{dz}{z}.
\end{multline}

\eqref{eq:alternative_zetaratios} is of course a different expression than \eqref{eq:zetaratios}. Thus for averages of ratios of the Riemann zeta-function, an ACUE spacing would be distinguished from CUE spacing. In fact using the methods of \cite{CoSn2,Ro} it should be possible to demonstrate rigorously that \eqref{eq:alternative_zetaratios} is equivalent to the AGUE Hypothesis, but we do not pursue this here.

\Addresses

\end{document}